\documentclass{llncs}
\usepackage{amssymb}
\usepackage{amsmath}
\usepackage{mathrsfs}
\usepackage{enumerate}
\usepackage{xspace}
\usepackage{graphicx}
\usepackage{subfigure}
\usepackage{url}
\usepackage{tikz}
\usetikzlibrary{automata,positioning}
\tikzstyle{every path}=[arrows=-latex]
\tikzstyle{every loop}=[arrows=-latex]
\tikzstyle{accepting}=[double distance=2pt]
\newcommand{\defAs}{\:=\!\!_{\mbox{\tiny Def}}\:}
\newcommand{\ceil}[1]{\lceil #1 \rceil}
\newcommand{\floor}[1]{\lfloor #1 \rfloor}
\newcommand{\AND}{\mathrel{\&}}

\newcommand{\NOT}{\mathop{\sim}}
\newcommand{\onelen}{\text{\tt 1-len}}
\newcommand{\id}{\text{\tt id}}
\newcommand{\bpstyle}[1]{\mathsf{#1}}

\title{On a compact encoding of the swap automaton}

\author{Kimmo Fredriksson\inst{1} \and Emanuele Giaquinta\inst{2}}
\institute{School of Computing, University of Eastern Finland
	  \email{kimmo.fredriksson@uef.fi} \and
          Department of Computer Science, University of Helsinki, Finland
	  \email{emanuele.giaquinta@cs.helsinki.fi}}

\begin{document}

\maketitle

\begin{abstract}
Given a string $P$ of length $m$ over an alphabet $\Sigma$ of size $\sigma$, a
swapped version of $P$ is a string derived from $P$ by a series of
local swaps, i.e., swaps of adjacent symbols, such that each symbol
can participate in at most one swap.
We present a theoretical analysis of the nondeterministic finite
automaton for the language $\bigcup_{P'\in\Pi_P}\Sigma^*P'$ (swap
automaton for short), where $\Pi_P$ is the set of swapped versions of
$P$. Our study is based on the bit-parallel simulation of the same
automaton due to Fredriksson, and reveals an interesting combinatorial
property that links the automaton to the one for the language $\Sigma^*P$.
By exploiting this property and the method presented by Cantone et al.
(2010), we obtain a bit-parallel encoding of the swap automaton which
takes $O(\sigma^2\ceil{k/w})$ space and allows one to simulate the
automaton on a string of length $n$ in time $O(n\ceil{k/w})$, where
$\ceil{m/\sigma}\le k\le m$.
\end{abstract}

\section{Introduction}

The \emph{Pattern Matching with Swaps} problem (Swap Matching problem,
for short) is a well-studied variant of the classic Pattern Matching
problem. It consists in finding all occurrences, up to character
swaps, of a pattern $P$ of length $m$ in a text $T$ of length $n$,
with $P$ and $T$ sequences of characters over a common finite
alphabet $\Sigma$ of size $\sigma$.
More precisely, the pattern is said to match the text at a given
location $j$ if adjacent pattern characters can be swapped, if
necessary, so as to make it identical to the substring of the text
ending (or, equivalently, starting) at location $j$. All swaps are
constrained to be disjoint, i.e., each character can be involved at
most in one swap.

The Swap Matching problem was introduced in 1995 as one of the open
problems in nonstandard string
matching~\cite{DBLP:dblp_conf/cpm/Muthukrishnan95}. The first result
that improved over the naive $O(nm)$-time bound is due to Amir et
al.~\cite{DBLP:journals/jal/AmirALLL00}, who presented an
$O(nm^\frac{1}{3}\log m)$-time algorithm for binary alphabets and
described how to reduce the case of a general alphabet to that of a
binary one with a $O(\log\sigma)$-time overhead. The best theoretical
result to date is due to Amir \emph{et
  al.}~\cite{DBLP:journals/iandc/AmirCHLP03}. Their algorithm runs in
time $O(n\log m)$ for binary alphabets and can also solve the case of
general alphabets in time $O(n\log m\log \sigma)$ by using again the
alphabet reduction technique of Amir \emph{et
  al.}~\cite{DBLP:journals/jal/AmirALLL00}. Both solutions are based
on reducing the problem to convolutions. Note that this problem can
also be solved using more general algorithms for Approximate String
Matching~\cite{DBLP:journals/csur/Navarro01}, albeit with worse
bounds.

There also exist different practical solutions, based on word-level
parallelism. To our knowledge, the first one is due to
Fredriksson~\cite{Fredriksson00}, who presented a
generalization of the nondeterministic finite automaton (NFA) for the
language $\Sigma^*P$ (prefix automaton) for the Swap Matching problem and a fast method
to simulate it using bit-parallelism~\cite{BYG92}. The resulting
algorithm runs in $O(n\ceil{m/w})$-time and uses $O(\sigma\ceil{m/w})$
space, where $w$ is the machine word size in bits. In the same paper
Fredriksson also presented a variant of the BNDM
algorithm~\cite{DBLP:journals/jea/NavarroR00}, based on the
generalization of the NFA for the language of the suffixes of $P$ (suffix automaton),
which achieves sublinear time on average and runs in
$O(nm\ceil{m/w})$-time in the worst-case. In $2008$ Iliopoulos and
Rahman presented a variant of Shift-Or for this problem, based on a
Graph-Theoretic model~\cite{DBLP:conf/sofsem/IliopoulosR08}. Their
algorithm runs in time $O(n\ceil{m/w}\log m)$ and uses
$O(m\ceil{m/w})$ space (the $\log m$ term can be removed at the price of $O(\sigma^2\ceil{m/w})$ space).
The improvement over the algorithm by
Fredriksson is that the resulting bit-parallel simulation is simpler,
in that it requires fewer bitwise operations.
Later, Cantone and Faro presented an algorithm based on dynamic
programming that runs in time $O(n\ceil{m/w})$ and requires
$O(\sigma\ceil{m/w})$ space~\cite{DBLP:conf/sofsem/CantoneF09}.
Subsequently Campanelli et al. presented a variant of the BNDM
algorithm based on the same approach which runs in
$O(nm\ceil{m/w})$-time in the
worst-case~\cite{DBLP:conf/iwoca/CampanelliCF09}.

In~\cite{DBLP:journals/iandc/CantoneFG12} Cantone et al. presented a
technique to encode the prefix automaton in $O(\sigma^2\ceil{k/w})$
space and simulate it on a string of length $n$ in $O(n\ceil{k/w})$
time, where $\ceil{m/\sigma}\le k\le m$. In this paper we extend this
result to the NFA described in~\cite{Fredriksson00}. First, we present
a theoretical analysis of this NFA, from which the correctness of the
bit-parallel simulation presented in the same paper follows. We then
show that, by exploiting the properties of this NFA that we reveal in
the following, we can solve the Swap Matching problem in time
$O(n\ceil{k/w})$ and space $O(\sigma^2\ceil{k/w})$, where
$\ceil{m/\sigma}\le k\le m$, using the method presented
in~\cite{DBLP:journals/iandc/CantoneFG12}. Our result also applies,
with small changes, to the case of the generalized suffix automaton
for the Swap Matching problem.

\section{Notions and Basic Definitions}
Given a finite alphabet $\Sigma$ of size $\sigma$, we denote by $\Sigma^{m}$, with $m
\geq 0$, the collection of strings of length $m$ over $\Sigma$ and put
$\Sigma^{*} = \bigcup_{m\in \mathbb{N}} \Sigma^{m}$.
We represent a string $P \in \Sigma^{m}$ as
an array $P[0\,..\,m-1]$ of characters of $\Sigma$ and write
$|P| = m$ (in particular, for $m=0$ we obtain the empty string
$\varepsilon$). Thus, $P[i]$ is the $(i+1)$-st character of $P$, for
$0\le i< m$, and $P[i\,\ldots\,j]$ is the substring of $P$ contained
between its $(i+1)$-st and $(j+1)$-st characters, inclusive, for $0\le i \le j <
m$.
For any two strings $P$ and $P'$, we write $PP'$ to denote the
concatenation of $P$ and $P'$.

Given a string $P\in \Sigma^m$, we indicate with
$\mathcal{A}(P)=(Q,\Sigma,\delta,q_0,F)$ the nondeterministic finite
automaton (NFA) for the language $\Sigma^*P$ of all words in $\Sigma^{*}$
ending with an occurrence of $P$ (prefix automaton for short), where:
\begin{itemize}
    \item
    $Q=\{q_0,q_1,\ldots, q_m\}$ \qquad ($q_{0}$ is the initial state)

    \item the transition function $\delta: Q \times \Sigma
    \longrightarrow \mathscr{P}(Q)$ is defined by:
    $$
    \delta(q_{i},c) \defAs \begin{cases}
    \{q_{0},q_{1}\} & \text{if } i = 0 \text{ and } c = P[0]\\
    \{q_{0}\} & \text{if } i = 0 \text{ and } c \neq P[0]\\
    \{q_{i+1}\} & \text{if } 1 \le i < m \text{ and } c = P[i]\\
    \emptyset & \text{otherwise}
    \end{cases}
    $$
    \item
    $F=\{q_m\}$ \qquad ($F$ is the set of final states).
\end{itemize}

The valid configurations $\delta^*(q_{0},S)$ which are reachable by the automaton
$\mathcal{A}(P)$ on input $S \in \Sigma^{*}$ are defined recursively as
follows:
$$
    \delta^*(q_{0},S)=_{\mathit{Def}}\begin{cases}
    \{q_0\} & \text{if $S=\varepsilon$,}\\
    \bigcup_{q' \in \delta^{*}(q_{0},S')}\delta(q',c) & \text{if
    $S=S'c$,
    for some $c\in \Sigma$ and $S' \in \Sigma^{*}$.}
    \end{cases}
$$

\begin{definition}
A \emph{swap permutation} for a string $P$ of length $m$ is a
permutation $\pi : \{0,...,m-1\} \rightarrow \{0,...,m-1\}$ such that:
\begin{enumerate}[(a)]
    \item if $\pi(i) = j$ then $\pi(j) = i$ (characters are swapped);

    \item for all $i$, $\pi(i) \in \{i-1, i, i+1\}$ (only adjacent
    characters are swapped);

    \item if $\pi(i) \neq i$ then $P[\pi(i)]
    \neq P[i]$ (identical characters are not swapped).
\end{enumerate}
\end{definition}

For a given string $P$ and a swap permutation $\pi$ for $P$, we write
$\pi(P)$ to denote the \emph{swapped version} of $P$, namely $\pi(P) =
P[\pi(0)]P[\pi(1)]\ldots P[\pi(m-1)]$.

\begin{definition}[Pattern Matching with Swaps Problem]
  Given a text $T$ of length $n$ and a pattern $P$ of length $m$, find
  all locations $j \in \{m-1,...,n-1\}$ for which there exists a swap
  permutation $\pi$ of $P$ such that $\pi(P)$ matches $T$ at location
  $j$, i.e. $P[\pi(i)] = T[j-m+i+1]$, for $i=0...m-1$.
\end{definition}

Finally, we recall the notation of some bitwise infix operators on
computer words, namely the bitwise \texttt{and} ``$\&$'', the bitwise
\texttt{or} ``$|$'', the \texttt{left shift} ``$\ll$'' operator (which
shifts to the left its first argument by a number of bits equal to its
second argument), and the unary bitwise \texttt{not} operator
``$\NOT$''.

\subsection{$1$-factorization encoding of the prefix automaton}
A $1$-factorization $\boldsymbol{u}$ of size $k$ of a string $P$ is a sequence
$\langle u_1,u_2,\ldots, u_k\rangle$ of nonempty substrings of $P$
such that:
\begin{enumerate}[(a)]
\item $P=u_1 u_2\dots u_k$\,;
\item each factor $u_{j}$ in $\boldsymbol{u}$ contains at most
\emph{one} occurrence of any of the characters in the alphabet
$\Sigma$, for $j=1,\ldots,k$\,.
\end{enumerate}
The following result was presented in~\cite{DBLP:journals/iandc/CantoneFG12}:
\begin{theorem}[cf.~\cite{DBLP:journals/iandc/CantoneFG12}]\label{thm:1-factorization}\label{thm:1-encoding}
  Given a string $P$ of length $m$ and a $1$-factorization of $P$ of length
  $\ceil{m/\sigma}\le k\le m$, we can encode the automaton
  $\mathcal{A}(P)$ in $O(\sigma^2\ceil{k/w})$ space and simulate it in
  time $O(n\ceil{k/w})$ on a string of length $n$.
\end{theorem}
We briefly recall how the encoding of Theorem~\ref{thm:1-encoding}
works. A $1$-factorization $\langle u_1,u_2,\ldots, u_k\rangle$ of $P$
induces a partition $\{Q_{1},\ldots,Q_{k}\}$ of the set $Q \setminus
\{q_{0}\}$ of states of the automaton $\mathcal{A}(P)$, where
$$
Q_{i} \defAs
\left\{q_{r_{i}+1},
\ldots,q_{r_{i+1}}\right\}\,, \text{ for }
i=1,\ldots,k\,,
$$
and $r_{j} = |u_{1}u_{2}\ldots u_{j-1}|$, for $j=1,\ldots,k+1$. We
denote with $q_{i,a}$ the unique state in $Q_i$ with an incoming
transition labeled by $a$, if such a state exists; otherwise $q_{i,a}$
is undefined.
The configuration $\delta^{*}(q_{0},Sa)$ of
$\mathcal{A}(P)$ on input $Sa$ can then be encoded by the pair
$(\bpstyle{D},a)$, where $\bpstyle{D}$ is the bit-vector of size $k$
such that $\bpstyle{D}[i]$ is set iff $q_{i,a} \in
\delta^{*}(q_{0},Sa)$.
We also denote with $\id(i, u_i[j]) = r_i + j$ the position of
symbol $u_i[j]$ in $P$, for $0\le j\le |u_i|-1$. Equivalently,
$\id(i,u_i[j])$ is the index of state $q_{i,u_i[j]}$ in the original
automaton.

\section{An analysis of the swap automaton}

Let $P$ be a pattern of length $m$ and let $\Pi_P$ be the set
including all the swapped versions of $P$. The swap automaton of $P$
is the nondeterministic finite automaton that recognizes all the words
in $\Sigma^*$ ending with a swapped version of $P$. Formally, it is
the NFA $\mathcal{A}_{\pi}(P)=(Q,\Sigma,\delta,q_0,F)$, where:
\begin{itemize}
\item $Q=\{q_0,q_1\ldots,q_{2m-1}\}$
    \item the transition function $\delta: Q \times \Sigma
    \longrightarrow \mathscr{P}(Q)$ is defined by:
    $$
    \delta(q_{i},c) \defAs \begin{cases}
    \{q_{0},q_{1}\} & \text{if } i = 0 \text{ and } c = P[0]\\
    \{q_{0},q_{m+1}\} & \text{if } i = 0 \text{ and } c = P[1]\\
    \{q_{0}\} & \text{if } i = 0 \text{ and } c \neq P[0] \text{ and } c \neq P[1]\\
    \{q_{i+1}\} & \text{if } 1 \le i < m \text{ and } c = P[i]\\
    \{q_{i+m+1}\} & \text{if } 1 \le i < m-1 \text{ and } c = P[i+1]\\
    \{q_{i-m+1}\} & \text{if } m+1 \le i < 2m \text{ and } c = P[i-m-1]\\
    \emptyset & \text{otherwise}
    \end{cases}
    $$
\item $F=\{q_m\}$
\end{itemize}
The language accepted by $\mathcal{A}_{\pi}(P)$ is
$\mathcal{L}(\mathcal{A}_{\pi}(P))=\bigcup_{P'\in\Pi_P}\Sigma^*P'$. An
example of this automaton for the string $cagca$ is depicted in
Fig.~\ref{fig:example}.
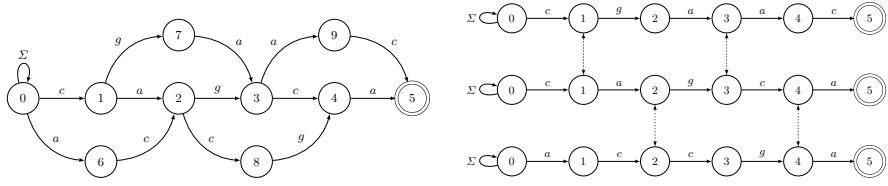
\begin{figure}[t]
\begin{center}
\subfigure{
\resizebox{0.47\textwidth}{!}{
\begin{tikzpicture}[node distance=2cm,auto]

\node[state] (q_0) {$0$};
\node[state] (q_1) [right of=q_0] {$1$};
\node[state] (q_2) [right of=q_1] {$2$};
\node[state] (q_3) [right of=q_2] {$3$};
\node[state] (q_4) [right of=q_3] {$4$};
\node[state,accepting] (q_5) [right of=q_4] {$5$};
\node[state] (q_6) [below=0.8cm of q_1] {$6$};
\node[state] (q_7) [above=0.8cm of q_2] {$7$};
\node[state] (q_8) [below=0.8cm of q_3] {$8$};
\node[state] (q_9) [above=0.8cm of q_4] {$9$};

\path[->] (q_0) edge node {$c$} (q_1)
                edge [loop above] node {$\Sigma$} ()
          (q_1) edge node {$a$} (q_2)
          (q_2) edge node {$g$} (q_3)
          (q_3) edge node {$c$} (q_4)
          (q_4) edge node {$a$} (q_5);

\path[->] (q_0) edge [bend right=35] node {$a$} (q_6)
          (q_6) edge [bend right=35] node {$c$} (q_2)
          (q_2) edge [bend right=35] node {$c$} (q_8)
          (q_8) edge [bend right=35] node {$g$} (q_4)
          (q_1) edge [bend left=35] node {$g$} (q_7)
          (q_7) edge [bend left=35] node {$a$} (q_3)
          (q_3) edge [bend left=35] node {$a$} (q_9)
          (q_9) edge [bend left=35] node {$c$} (q_5);

\end{tikzpicture}
}
}
\subfigure{
\resizebox{0.47\textwidth}{!}{
\begin{tikzpicture}[node distance=2cm,auto]
\node[state] (q_0) {$0$};
\node[state] (q_1) [right of=q_0] {$1$};
\node[state] (q_2) [right of=q_1] {$2$};
\node[state] (q_3) [right of=q_2] {$3$};
\node[state] (q_4) [right of=q_3] {$4$};
\node[state,accepting] (q_5) [right of=q_4] {$5$};

\path[->] (q_0) edge node {$c$} (q_1)
                edge [loop left] node {$\Sigma$} ()
          (q_1) edge node {$a$} (q_2)
          (q_2) edge node {$g$} (q_3)
          (q_3) edge node {$c$} (q_4)
          (q_4) edge node {$a$} (q_5);

\node[state] (a_0) [above of=q_0] {$0$};
\node[state] (a_1) [right of=a_0] {$1$};
\node[state] (a_2) [right of=a_1] {$2$};
\node[state] (a_3) [right of=a_2] {$3$};
\node[state] (a_4) [right of=a_3] {$4$};
\node[state,accepting] (a_5) [right of=a_4] {$5$};

\path[->] (a_0) edge node {$c$} (a_1)
                edge [loop left] node {$\Sigma$} ()
          (a_1) edge node {$g$} (a_2)
          (a_2) edge node {$a$} (a_3)
          (a_3) edge node {$a$} (a_4)
          (a_4) edge node {$c$} (a_5);

\node[state] (b_0) [below of=q_0] {$0$};
\node[state] (b_1) [right of=b_0] {$1$};
\node[state] (b_2) [right of=b_1] {$2$};
\node[state] (b_3) [right of=b_2] {$3$};
\node[state] (b_4) [right of=b_3] {$4$};
\node[state,accepting] (b_5) [right of=b_4] {$5$};

\path[->] (b_0) edge node {$a$} (b_1)
                edge [loop left] node {$\Sigma$} ()
          (b_1) edge node {$c$} (b_2)
          (b_2) edge node {$c$} (b_3)
          (b_3) edge node {$g$} (b_4)
          (b_4) edge node {$a$} (b_5);

\path[->] (q_1) edge [dotted] node {} (a_1)
          (a_1) edge [dotted] node {} (q_1)
          (q_2) edge [dotted] node {} (b_2)
          (b_2) edge [dotted] node {} (q_2)
          (q_3) edge [dotted] node {} (a_3)
          (a_3) edge [dotted] node {} (q_3)
          (q_4) edge [dotted] node {} (b_4)
          (b_4) edge [dotted] node {} (q_4);
\end{tikzpicture}
}
}
\caption{(a) The swap automaton for the pattern $cagca$; (b) The decomposition of the swap automaton for the pattern $cagca$.}
\label{fig:example}
\end{center}
\end{figure}
Compared to the NFA $\mathcal{A}(P)$ for the language $\Sigma^*P$, this
automaton has $m-1$ additional states and $2m-2$ additional
transitions. To our knowledge, this automaton was described for the
first time by Fredriksson in~\cite{Fredriksson00}. In the same paper, Fredriksson
presented an efficient simulation of this automaton based on
word-level parallelism. Let $\phi(S)=S'$ be the string of length $|S|$ defined as follows:
$$
S'[i]=
\begin{cases}
S[i] & \text{if } i \ge \floor{m/2}2 \\
S[i-1] & \text{if } i\bmod 2 = 1 \\
S[i+1] & \text{if } i\bmod 2 = 0 \\
\end{cases}
$$
The method is based on the decomposition of the swap automaton for $P$
into the three automata $\mathcal{A}(P)$, $\mathcal{A}(P_e)$ and
$\mathcal{A}(P_o)$, where $P_e=P[0]\phi(P[1\,\ldots\,m-1])$ and
$P_o=\phi(P)$. In the case of the string $cagca$ we have that $P_e =
cgaac$ and $P_o=accga$. The corresponding automata are depicted in
Fig.~\ref{fig:example}. Observe that all the automata have exactly
$m+1$ states. We denote with $q^1_i$, $q^2_i$ and $q^3_i$ the $i$-th
state of the automata $\mathcal{A}(P)$, $\mathcal{A}(P_e)$ and
$\mathcal{A}(P_o)$, respectively. Likewise for the corresponding
transition functions. Given a string $S$, let $D^i_j(S)$ be the set recursively defined as
$$D^i_j(S) =
\begin{cases}
\bigcup_{q\in D^i_{j-1}(S)\cup C^i_{j-1}(S)}\delta_i(q, S[j]) & \text{if } 1\le j\le |S|-1 \\
\delta_i(q^i_0,S[0]) & \text{if } j = 0.
\end{cases}
$$
for $i=1,\ldots,3$, where
$$
C^1_j(S) = \{q^1_i\ |\ (i\bmod 2 = 1\wedge q^2_i\in D^2_j(S))\vee (i\bmod 2 = 0\wedge q^3_i\in D^3_j(S))\}\,,
$$
$$
C^2_j(S) = \{q^2_i\ |\ (i\bmod 2 = 1\wedge q^1_i\in D^1_j(S))\}\,,
$$
$$
C^3_j(S) = \{q^3_i\ |\ (i\bmod 2 = 0\wedge q^1_i\in D^1_j(S))\}\,,
$$
for $j=0,\ldots,|S|-1$.
The idea is to simulate the three automata simultaneously
on $S$. However, at each iteration, we also activate some states of
each automaton depending on the configuration of the others. More
precisely, we activate state $q^2_i$ if $i$ is odd and state $q^1_i$
is active, and viceversa. Similarly, we activate state $q^3_i$ if $i$
is even and state $q^1_i$ is active, and viceversa.
The sets $D^i_j$ encode the described configurations.
Now, consider the automaton $\mathcal{A}_{\pi}(P)$. It is not hard to
see that the following Lemma holds:

\begin{lemma}\label{lemma:swap-prop}
  In a simulation of the automaton $\mathcal{A}_{\pi}(P)$ on a given
  string $S$, state $q_i$ is active at the $j$-th
  iteration, i.e., $q_i\in \delta^*(q_0, S[0\,\ldots\,j])$, iff one of
  the following three conditions hold:
\begin{enumerate}
\item $1\le i\le m$ and $q^1_i\in D^1_j(S)$;
\item $m < i < 2m$, $i-m$ is even and $q^2_{i-m}\in D^2_j(S)$;
\item $m < i < 2m$, $i-m$ is odd and $q^3_{i-m}\in D^3_j(S)$.
\end{enumerate}
\end{lemma}
Hence, to simulate $\mathcal{A}_{\pi}(P)$ it is enough to simulate the
automata $\mathcal{A}(P)$, $\mathcal{A}(P_e)$ and $\mathcal{A}(P_o)$,
and compute the sets $D^i_j$. To this end,
Fredriksson uses the well known technique of
bit-parallelism~\cite{BYG92} to encode each automaton in $O(\sigma\ceil{m/w})$ space.
For a given string $T$ of length $n$, the simulation
of the three automata on $T$ can be then computed in time
$O(n\ceil{m/w})$, since the number of automata is constant. For the
details concerning the bit-parallel simulation
see~\cite{Fredriksson00}.

We now show how to exploit Lemma~\ref{lemma:swap-prop} to devise an improved
algorithm for the Swap Matching problem.
Our result will be a combination of Lemma~\ref{lemma:swap-prop} and
Theorem~\ref{thm:1-factorization}. The idea is to encode each
automaton using a $1$-factorization of the corresponding string.
However, for the simulation to work, we must be able to compute the sets $C^i_j(S)$ in constant time (per word),
which is not trivial using the $1$-factorization encoding. The first
prerequisite for a constant time computation is the following property:
\begin{property}\label{prop:bit}
For any pair of states $(q^1_i,q^2_i)$ or $(q^1_i,q^3_i)$, the two
states in the pair map onto the same bit position in the bit-vector
encoding of the corresponding automaton.
\end{property}
For this to hold, given a sequence of factorizations
$\boldsymbol{u}^1$, $\boldsymbol{u}^2$, \ldots,
$\boldsymbol{u}^{\ell}$ we must have that
\begin{enumerate}
\item $|\boldsymbol{u}^i| = |\boldsymbol{u}^j|$, for any $1\le i,j\le \ell$
\item $|u^i_l| = |u^j_l|$, for any $1\le i,j\le \ell$ and $1\le l\le |\boldsymbol{u}^i|$
\end{enumerate}
These conditions are not satisfied in general by the minimal
$1$-factorizations of the strings. For example, the minimal
$1$-factorizations of $cagca$, $cgaac$ and $accga$ are $\langle cag,
ca\rangle$, $\langle cga, ac\rangle$ and $\langle ac, cga\rangle$, and
the last factorization does not satisfy condition $2$.
Let $\onelen(S, s) = i$, where $i$ is the length such that all the
symbols in $S[s\,\ldots\,s+i-1]$ are distinct and either $s+i-1=|S|-1$ or
$S[s+i]$ occurs in $S[s\,\ldots\,s+i-1]$, for $s=0,\ldots,|S|-1$.
We introduce the following definition:
\begin{definition}
Given a sequence $\mathcal{S}$ of strings $S_1$, $S_2$, \ldots,
$S_{\ell}$ of the same length, we define the $1$-collection of
$\mathcal{S}$ as the sequence of $1$-factorizations
$\boldsymbol{u}^1, \boldsymbol{u}^2,\ldots,\boldsymbol{u}^{\ell}$ of
length $k$, where $\boldsymbol{u}^i=\langle u^i_1, u^i_2, \ldots
u^i_k\rangle$, such that
\begin{enumerate}[(a)]
\item $S_i = u^i_1u^i_2\ldots u^i_k$;
\item $|u^i_j| = \min\limits_{S\in\mathcal{S}}\onelen(S, \sum_{l=1}^{j-1} |u^i_l|)$.
\end{enumerate}
Observe that the $1$-collection of $\mathcal{S}$ satisfies conditions $1$ and $2$.
\end{definition}
For example, the $1$-collection of $cagca$, $cgaac$ and $accga$ is
$\langle ca, g, ca\rangle$, $\langle cg, a, ac\rangle$ and $\langle
ac, c, ga\rangle$. Indeed, we can encode the automata
$\mathcal{A}(P)$, $\mathcal{A}(P_e)$ and $\mathcal{A}(P_o)$ using
Theorem~\ref{thm:1-factorization} and the $1$-collection of $P, P_e,
P_o$ in space $O(\sigma^2\ceil{k/w})$, where $k$ is the size of any
$1$-factorization in the $1$-collection of $P, P_e, P_o$. By
definition, the $1$-collection of $P, P_e, P_o$ satisfies conditions $1$ and $2$, and thus
Property~\ref{prop:bit} holds.

Before continuing, we first bound the size $k$ of the
factorizations in the $1$-collection of $P$, $P_e$ and $P_o$.
\begin{lemma}
  Let $k'$ be the size of a minimal $1$-factorization of $P$ and let
  $k$ be the size of any factorization in the $1$-collection of $P$,
  $P_e$ and $P_o$. Then we have $k\le \min(3 k' - 2, m)$.
\end{lemma}
\begin{proof}
Let $\langle u_1,u_2,\ldots,u_{k'}\rangle$ be the (greedy) minimal
$1$-factorization of $P$ such that $|u_j| = \onelen(P,
\sum_{l=1}^{j-1} |u_l|)$, for $j=1,\ldots,k'$. Let $s =
\sum_{l=1}^{j-1}|u_l|$ for a given $j$, and suppose that
$|u_j|=\onelen(P, s) = i$, so that $P[s+i]$ occurs in
$P[s\,\ldots\,s+i-1]$. If $s+i$ is even, then $P_e[s+i-1] = P[s+i]$ and
$P_o[s+i+1] = P[s+i]$; viceversa if $s+i$ is odd. Suppose that $s+i$
is even (the other case is analogous).

If $s$ is even then $P_o[s\,\ldots\,s+i-1]$ is a permutation of
$P[s\,\ldots\,s+i-1]$, which implies $\onelen(P_o, s)\ge i$.
Instead, in the case of $P_e$, $P_e[s+1\,\ldots\,s+i-2]$ is a
permutation of $P[s+1\,\ldots\,s+i-2]$. This implies that
$\onelen(P_e, s+1) \ge i-2$.

If $s$ is odd then $P_e[s\,\ldots\,s+i-2]$ is a permutation of
$P[s\,\ldots\,s+i-2]$, which implies that $\onelen(P_e, s)\ge i-1$.
Instead, in the case of $P_o$, $P_o[s+1\,\ldots\,s+i-1]$ is a
permutation of $P[s+1\,\ldots\,s+i-1]$. This implies that
$\onelen(P_e, s+1) \ge i-1$.

Observe that $\onelen(S,s)\ge i$ implies $\onelen(S,s+1)\ge i-1$.
In both cases, we assume pessimistically that
$
\min_{S\in\{P,P_e,P_o\}}\onelen(S, s)=1\,,
$
$
\min_{S\in\{P,P_e,P_o\}}\onelen(S, s+1)=i-2\,,
$
and
$
\min_{S\in\{P,P_e,P_o\}}\onelen(S, s+i-1)=1\,.
$
This arrangement is compatible with the constraints described above.

In this way each factor $u_j$ covers three factors in the
$1$-collection of $P$, $P_e$ and $P_o$. However, a finer analysis
reveals that $u_1$ and $u_k$ can cover two factors only.
Indeed, in the case of $u_1$ we have that $P_e[0\,\ldots\,i-2]$ is a
permutation of $P[0\,\ldots\,i-2]$ and $P_o[0\,\ldots\,i-1]$ is a
permutation of $P[0\,\ldots\,i-1]$, so we can assume
$\min_{S\in\{P,P_e,P_o\}}\onelen(S, 0)=i-1$ and
$\min_{S\in\{P,P_e,P_o\}}\onelen(S, i-1)=1$. Instead, in the case of
$u_k$ we have that $P_o[s\,\ldots\,m-1]$ is a permutation of
$P[s\,\ldots\,m-1]$ and $P_e[s+1\,\ldots\,m-1]$ is a permutation of
$P[s+1\,\ldots\,m-1]$, if $s$ is even, viceversa if $s$ is odd. So we
can assume $\min_{S\in\{P,P_e,P_o\}}\onelen(S, s)=1$ and
$\min_{S\in\{P,P_e,P_o\}}\onelen(S, s+1)=m-s-1$. The claim then
follows.

\qed
\end{proof}

We now describe a property of the $1$-collection of strings $P$, $P_e$
and $P_o$ that will be the key for the constant time computation of
$C^i_j$:
\begin{lemma}\label{lemma:mapping}
  Let $\boldsymbol{u}^1$, $\boldsymbol{u}^2$ and $\boldsymbol{u}^3$ be
  the $1$-collection of $P$, $P_e$ and $P_o$. Then, the following
  facts hold:
\begin{itemize}
\item $\id^1(i,P[j]) = \id^2(i,P[j-1])$ \\$\id^2(i,P_e[j]) = \id^1(i,P_e[j-1])$ if $j\bmod 2 = 0$
\item $\id^1(i,P[j]) = \id^3(i,P[j-1])$ \\$\id^3(i,P_o[j]) = \id^1(i,P_o[j-1])$ if $j\bmod 2 = 1$
\end{itemize}
for $1\le i\le k$ and $\max(r^1_i, 1)\le j\le r^1_i+|u^1_i|-1$.
\end{lemma}
\begin{proof}
  By definition, $id^1(i,P[j]) = j$, for any $i,j$ as above, since
  $\boldsymbol{u}^1$ is a factorization of $P$. Similarly,
  $id^2(i,P_e[j]) = j$, since $r^2_i=r^1_i$ and $|u^2_i|=|u^1_i|$. If
  $j\bmod 2 = 0$, then $P_e[j] = P[j-1]$ and $P[j] = P_e[j-1]$ so that
  $id^2(i,P_e[j]) = id^2(i,P[j-1])$ and $id^1(i,P[j]) =
  id^1(i,P_e[j-1])$. The case of $j\bmod 2 = 1$ is analogous with
  $P_o$ and $id^3$ in place of $P_e$ and of $id^2$, respectively.

\qed
\end{proof}
We now present how to compute the sets $C^2_j$ and $C^3_j$. The case of $C^1_j$ 
 is analogous. More precisely, we need to compute the
$1$-factorization encoding of $C^2_j$ and $C^3_j$, given the pair
$(\bpstyle{D}^1,T[j])$ encoding the set $D^1_j(T)$.
Let $\bpstyle{E}(c)$ be a bit-vector of $k$ bits such that bit $i$
is set in $\bpstyle{E}(c)$ iff $\id^1(i,c)$ is even, for any
$c\in\Sigma$.
First we compute the bit-vector $\bpstyle{D}'$ such that bit $i$ is set iff
bit $i$ is set in $\bpstyle{D}^1$ and $\id^1(i,T[j])$ is
even. This can be done in constant time by performing a bitwise and of
$\bpstyle{D}^1$ with $\bpstyle{E}(T[j])$. Observe that the
pair $(\bpstyle{D}',T[j])$ encodes the set $ \{q^1_i\ |\ (i\bmod 2 =
1\wedge q^1_i\in D^1_j(T))\} $. We claim that the pair
$(\bpstyle{D}',T[j-1])$ encodes the set $C^2_j$. This follows by
Lemma~\ref{lemma:mapping} by observing that if bit $i$ is set in
$\bpstyle{D}'$ then $\id^1(i,T[j]) = \id^2(i,T[j-1])$.
The case of $C^3_j$ is symmetric, i.e., the pair $(D^1\AND \NOT\bpstyle{E}(T[j]), T[j-1])$ encodes $C^3_j$.

Given a string of length $n$, we can then simulate the swap automaton
using Lemma~\ref{lemma:swap-prop} in time $O(n\ceil{k/w})$. Hence, we
obtain the following result:
\begin{theorem}
  Given a string $P$ of length $m$, we can encode the automaton $\mathcal{A}_{\pi}(P)$ in
  $O(\sigma^2\ceil{k/w})$ space, where
  $\ceil{m/\sigma}\le k\le m$, and simulate it in time
  $O(n\ceil{k/w})$ on a string of length
  $n$.
\end{theorem}

\bibliographystyle{plain}
\bibliography{swaps}

\end{document}